\documentclass[aps,superscriptaddress,twocolumn,showpacs]{revtex4-2}

\usepackage{amsmath}
\usepackage{latexsym}
\usepackage{amssymb}
\usepackage{graphicx}
\usepackage[colorlinks=true, citecolor=blue, urlcolor=blue]{hyperref}
\usepackage{float}
\usepackage{amsfonts}
\usepackage{textcomp}
\usepackage{mathpazo}
\usepackage{comment}
\usepackage{blindtext}
\usepackage[dvipsnames]{xcolor}

\usepackage{bbm}

\usepackage{xcolor}
\definecolor{myurlcolor}{rgb}{0,0,0.4}
\definecolor{mycitecolor}{rgb}{0,0.5,0}
\definecolor{myrefcolor}{rgb}{0.5,0,0}
\usepackage{hyperref}
\hypersetup{colorlinks,
linkcolor=myrefcolor,
citecolor=mycitecolor,
urlcolor=myurlcolor}

\usepackage[draft]{fixme}
\usepackage{amsmath,bbm}
\usepackage{scrextend}
\usepackage{graphicx}
\usepackage{amsfonts}
\usepackage{amssymb}
\usepackage{amsmath,amssymb,amsthm,verbatim,graphicx,bbm}
\usepackage{mathrsfs}
\usepackage{mathtools}
\usepackage{color,xcolor,longtable}

\def\be{\begin{equation}}
\def\ee{\end{equation}}
\def\ben{\begin{eqnarray}}
\def\een{\end{eqnarray}}
\def\eea{\end{array}}
\def\bea{\begin{array}}

\newcommand{\Tr}[1]{\mathrm{Tr}#1}
\newcommand{\bei}{\begin{itemize}}
\newcommand{\eei}{\end{itemize}}
\newcommand{\ket}[1]{|#1\rangle}
\newcommand{\bra}[1]{\langle#1|}

\newcommand{\proj}[1]{\ket{#1}\!\bra{#1}}

\newcommand{\I}{\mathbbm{1}}

\renewcommand{\emph}[1]{\textbf{#1}}

\makeatletter
\newtheorem*{rep@theorem}{\rep@title}
\newcommand{\newreptheorem}[2]{%
\newenvironment{rep#1}[1]{%
 \def\rep@title{#2 \ref{##1}}%
 \begin{rep@theorem}}%
 {\end{rep@theorem}}}
\makeatother

\theoremstyle{plain}
\newtheorem{thm}{Theorem}
\newtheorem*{thm*}{Theorem}
\newreptheorem{thm}{Theorem}

\newtheorem{fakt}{Fact}

\theoremstyle{definition}

\theoremstyle{remark}

\usepackage[T1]{fontenc}

\begin{document}


\title{Any gate of a quantum computer can be certified device-independently}
\author{Shubhayan Sarkar}
\email{shubhayan.sarkar@ulb.be}
\affiliation{Laboratoire d’Information Quantique, Université libre de Bruxelles (ULB), Av. F. D. Roosevelt 50, 1050 Bruxelles, Belgium}
\affiliation{Institute of Informatics, Faculty of Mathematics, Physics and Informatics,
University of Gdansk, Wita Stwosza 57, 80-308 Gdansk, Poland}

\begin{abstract}	
Device-independent (DI) certification allows the verification of quantum systems based solely on observed statistics, without assumptions about their internal structure. While self-testing, the strongest DI certification, of a wide range of quantum states and measurements is done, the self-testing of quantum operations remains underdeveloped. Here, we show in a proof-of-principle way that any quantum unitary can be self-tested within the DI paradigm. For our purpose, we utilise the framework of quantum networks with multiple independent sources. Our work provides a fundamental step toward certifying quantum interactions directly from data, without detailed modelling assumptions. Moreover, the result also serves as a crucial ingredient for quantum computation, where verifying that quantum gates perform as intended is essential for building secure and reliable quantum processors.
\end{abstract}


\maketitle


\textit{Introduction.---} Devices based on quantum particles can do fascinating things like solving problems much faster than traditional computers and making communication more secure than ever. However, before we can reliably use them, we need to ensure that they behave in an intended manner. A way to do this is to either know the internal workings of the device or certify them. Traditionally, this meant trusting everything about how the device was built and how it works inside, which, given the current state of counterfeiting and cheating, is not at all realistic. Device-independent (DI) certification takes an alternative route which allows one to confirm the inner workings of a quantum device just by looking at the outputs it produces, without needing any additional information about the device apart from some minimal physical assumptions. The key ingredient for DI certification is Bell nonlocality \cite{Bell,Bell66}.

The strongest DI certification scheme is known as self-testing. Introduced in \cite{Mayers_selftesting, Yao}, self-testing enables unique certification of both quantum states and measurements, up to some degrees of freedom, using only the observed data, assuming nothing about the inner workings of the device except that it obeys the laws of quantum physics. In recent times, enormous effort has been devoted to certify quantum states, for instance, see Refs. \cite{Scarani,Reichardt_nature,Mckague_2014,Wu_2014,Bamps,All,chainedBell, Projection,Jed1,sarkar,sarkaro2,Allst,sarkarPRL} and quantum measurements \cite{Jed1, random1, sarkar, chen1, Marco, JW2, NLWEsupic, sarkarPRL, sarkar2024universal}. In particular, Ref. \cite{Allst} proposes a quantum-networks-based scheme to self-test any pure multipartite entangled state. An indirect way to self-test a mixed entangled state, in particular, a bound entangled state, was recently proposed in Ref. \cite{sarkarPRL}. This result was then generalised to any state and any measurement in Ref. \cite{sarkar2024universal}. 

While self-testing of quantum states and measurements has seen significant progress, self-testing quantum operations remains in its infancy, with only a handful of results to date. For instance, schemes to self-test single unitaries \cite{Dall} and a few entangling unitary gates have been proposed in \cite{pavel, sarkarint}. In this work, we introduce the first general DI scheme to self-test any quantum unitary, marking a significant step beyond earlier efforts focused on specific gates. Since any quantum operation can be represented as a unitary transformation on an extended Hilbert space, self-testing any unitary allows us to certify any interaction between quantum systems in a DI way. This offers a promising alternative to conventional methods of modelling quantum processes, such as Hamiltonian-based approaches in high-energy physics, which often rely on detailed assumptions about the involved particles. From a practical standpoint, this is especially relevant to quantum computing, where every gate is a unitary operation. Before building a quantum processor, it is crucial to verify that the individual gates perform as intended. A DI certification of these unitaries would enable such verification without the need to trust the gate manufacturer, which might serve as a major step toward secure and reliable quantum hardware. This is in contrast to the approaches proposed in \cite{Reichardt2013,Fitzsimons_2018} that self-test quantum computations.


For our purpose, we consider the framework of quantum networks where one has access to multiple independent sources. Utilising this framework, we provide proof-of-principle DI certification, or self-testing, of any quantum unitary. 
We first provide an almost device-independent way to certify any unitary gate. The idea of almost device-independence was introduced in \cite{sarkar3}, where one does not fully trust any of the devices involved in the setup but assumes that the local support of the states remains invariant. Here, we assume that the unitary gate does not change the local support of the state. The setup is based on a recent work \cite{sarkar2024universal}, which was used to certify any quantum state and measurement. Its simple structure provides an intuitive understanding of the certification of a quantum gate. Introducing additional components in this setup, inspired by \cite{Allst, sarkar2024steer}, we can drop the assumption of the invariance of local supports and finally provide a DI certification scheme for any quantum gate.   

Let us begin by describing the simple setup for almost DI certification of any quantum gate. 

\textit{Almost DI certification of any quantum gate---} To certify any unitary quantum operation we consider a quantum network as depicted in Fig. \ref{fig1}. The quantum network consists of $N$ external parties, denoted by parties $A_i$, with $i=1,\ldots, N$, and a central party $E$, and another external party $L$. The $N$ independent sources, denoted by $P_i$, in the network distribute bipartite quantum states among the parties such that each $P_i$ sends one subsystem to $A_i$ and other to $E$. We denote these states as $\rho_{A_i\overline{A}_i}$. 
For convenience, at times we will denote $A\coloneqq A_1\ldots A_N$ and $L\coloneqq\overline{A}_1\ldots\overline{A}_N$. 
On their shares of the joint state, each external party $A_i$ performs three binary measurements, where the inputs and outputs are denoted as $x_i=0,1,2$ and $a_i=0,1$, respectively. 

The central party $E$ receives $N$ subsystems from $P_i$ on which it performs an operation $V$ that can transform the incoming state. However, she can also freely choose not to interact with the incoming states. Consequently, her operation box takes two inputs $e=0,1$, where $e=0$ corresponds to $\I$ and $e=1$ corresponds to $V$. Finally at $L$, a single fixed $2^N-$outcome measurement is performed on the transformed states received from $E$, with the outputs denoted as $l=0,1,\ldots,2^N-1$. None of the parties communicate with each other during the experiment.

After the experiment, the parties reconstruct the probability distribution $\vec{p}=\{p(a_1\ldots a_Nl|x_1\ldots x_Ne)\}$, such that $p(a_1\ldots a_Nl|x_1\ldots x_Ne)$ is the probability of obtaining $\{a_i\}$ by $A_i's$ and $l$ by $L$ after they perform the measurements $\{x_i\}$ and $E$ chooses the transformation $e$, respectively. From Born's rule, we can write
\begin{equation}\label{probs}
    p(a_1\ldots a_Nl|x_1\ldots x_N e)=\Tr\left(\rho_{A\overline{A}} M_{a_i|x_i} \otimes V_e M_{l} V_e^{\dagger}\right),
\end{equation}
where $M_{i,x_i}=\{M_{a_i|x_i}\}$, for every $i$, and $M_l$ are the measurements performed by parties $A_i$ and $L$ with $V_{0}=\I,\ V_1=V$ denoting the transformations by Eve on the subsystem $\rho_{\overline{A}_1\ldots\overline{A}_N}$. Here, $\rho_{A\overline{A}}=\bigotimes_{i=1}^{N}\rho_{A_i\overline{A}_i}$.The measurement elements are positive semi-definite and satisfy $\sum_{a_i}M_{a_i|x_i}=\I_{A_i}$ for every measurement choice $x_i$ and every party $i$. The same applies to $L's$ measurements. Here, we also express the correlations in terms of expectation values of observables $\mathcal{A}_{i,x_i}$ as 
\begin{align}\label{obspic}
    &\langle \mathcal{A}_{1,x_1}\ldots \mathcal{A}_{N,x_N} V_eM_{l}V_e^{\dagger} \rangle=\Tr[(\bigotimes_{i=1}^{N} \rho_{A_i\overline{A}_i}\mathcal{A}_{i,x_i}) \otimes V_e^{\dagger}M_{l}V_e]\nonumber\\
    & = \sum_{a_1,\ldots,a_N=0,1} (-1)^{\sum_{i=1}^N a_i} p(a_1\ldots a_Nl|x_1\ldots x_N e).
\end{align}
Consequently, $\mathcal{A}_{i,x_i}=M_{0|x_i}-M_{1|x_i}$ for every $i=1,\ldots,N$. If all $M_{a_i|x_i}$ are projectors, then $A_{i,x_i}$ is unitary.

\begin{figure}[t!]
    \centering
    \includegraphics[width=\linewidth]{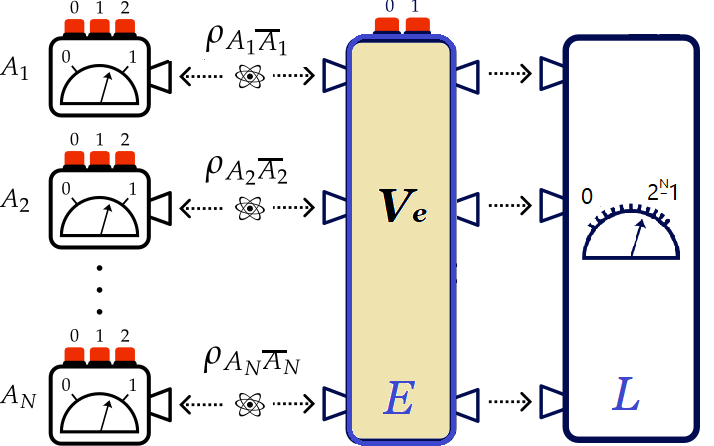}
    \caption{\textbf{Almost DI certification of any quantum gate.} The setup involves \( N+2 \) parties: \( A_i \) for \( i = 1, \ldots, N \), a central party \( E \) and $L$. There are \( N \) independent sources, each distributing a bipartite quantum state \( \rho_{A_i\overline{A}_i} \) between party \( A_i \) and \( E \). Each \( A_i \) has three possible inputs and two outcomes. The central party \( E \) has two inputs that corresponds to two different transformations $V_e$ such that $V_0=\I, V_1=V$, while $L$ performs a single \( 2^N \) outcome measurement.
}
    \label{fig1}
\end{figure}

As described above, the source $P_i$ generates the state $\rho_{A_i\overline{A}_i}$ acting on some Hilbert space $\mathcal{H}_{A_i}\otimes\mathcal{H}_{\overline{A}_i}$. The dimension of these spaces is unrestricted which allows us to consider that $\rho_{A_i\overline{A}_i}$ is pure. Moreover, we also recall that the measurements can only be certified on the support of the local states, and thus, without loss of generality, we take the measurements to be full-rank. In this section, we assume that the unitary $V$ does not change the local support of the states, that is, $V:\mathcal{H}_{\overline{A}_i}\rightarrow \mathcal{H}_{\overline{A}_i}$.  

Suppose now that there are some reference states $\ket{\psi'_{A_i\overline{A}_i}}\in \mathcal{H}_{A'_i}\otimes\mathcal{H}_{\overline{A}'_i}$, measurements $\{\mathcal{M}_{a_i|x_i}\}, \{\mathcal{M}_{l}\}$ and transformation $\mathcal{U}$ also generate $\vec{p}$ [Eq. \eqref{probs}]. We say that these quantum realisations are self-tested by observing $\vec{p}$ in the scenario depicted in Fig. \ref{fig1} if there are local unitary operations  $U_{A_i}:\mathcal{H}_{A_i}\to \mathcal{H}_{A_i'}\otimes\mathcal{H}_{A_i''}$ and $U_{\overline{A}_i}:\mathcal{H}_{\overline{A}_i}\to\mathcal{H}_{\overline{A}_i'}\otimes\mathcal{H}_{\overline{A}_i''}$ such that
\begin{equation}\label{ststate}
 (U_{A_i}\otimes U_{\overline{A}_i})\ket{\psi_{A_i\overline{A}_i}}=\ket{\psi'_{A_i'\overline{A}_i'}}\otimes\ket{\xi_{A_i''\overline{A}_i''}},
\end{equation}
where $\ket{\xi_{A_i''\overline{A}_i''}}\in\mathcal{H}_{A_i''}\otimes\mathcal{H}_{\overline{A}_i''}$ and
\begin{equation}\label{stmea}
U_{A_i}\,M_{i,x_i}\,U_{A_i}^{\dagger}=\mathcal{M}_{i,x_i}\otimes\mathbbm{1}_{A_i''},\quad U_{L}\, M_l\,U_{L}^{\dagger}=\mathcal{M}_l\otimes\I_{L''},
\end{equation}
where $\mathbbm{1}_{A_i''},\mathbbm{1}_{L''}$ is the identity acting on the auxiliary systems $\mathcal{H}_{A_i''},\mathcal{H}_{\overline{A}_i''}$ respectively and $U_{L}=\bigotimes_{i=1}^N U_{\overline{A}_i}$. Notice that we denote $L''=\overline{A}_1''\ldots\overline{A}_N''$. Moreover, the transformation $V$ is certified given some unitary $U$ as
\begin{eqnarray}\label{stU}
    U V U^{\dagger}=\mathcal{U}\otimes\I_{\mathrm{aux}}.
\end{eqnarray}
The above definitions are up to the complex conjugation of reference states, measurements and transformation. However, as we will see below, in this work we have the additional benefit that either the entire set of states, measurements and transformations is certified to be the reference ones or its complex conjugate.
Let us now proceed towards the self-testing of any quantum gate under the assumption that the transformation do not change the support of the local states. The scheme is divided into two parts, corresponding to two different inputs of $E$. 

{\it{Step 1.}} In the first step, we consider correlations $\{p(a_1\ldots a_Nl|x_1\ldots x_N 0)\}$. As $e=0$ corresponds to Eve not interacting with the system, we immediately follow \cite{sarkar2024universal} to self-test the sources and measurements of $A_i$ and $L$. For completeness, we give a brief overview of the scheme in \cite{sarkar2024universal}. First, using the correlations $\{p(a_1\ldots a_Nl|x_1\ldots x_N 0)\}$, a family of $2^N$ Bell inequalities are constructed as $\langle\mathcal{I}_{l}\rangle\leq\beta_C$ where the Bell operators $\mathcal{I}_{l}$ are given by
\begin{eqnarray}
\label{BE1Nm}
\mathcal{I}_{l}=(-1)^{l_1} (N-1)\tilde{\mathcal{A}}_{1,1}\bigotimes_{i=2}^N \mathcal{A}_{i,1} +\sum_{i=2}^N(-1)^{l_i}\tilde{\mathcal{A}}_{1,0}\mathcal{A}_{i,0}\nonumber\\-(-1)^{l_1+l_i}\sum_{i=2}^N  \mathcal{A}_{1,2}\mathcal{A}_{i,2}\bigotimes_{\substack{j=2\\j\ne i}}^{N} \mathcal{A}_{j,1}  \qquad
\end{eqnarray}
where $l\equiv l_1l_2\ldots l_N$ such that $l_1,l_2,\ldots,l_N=0,1$ and
\begin{equation}\label{overAm}
    \tilde{\mathcal{A}}_{1,0}=\frac{(\mathcal{A}_{1,0}-\mathcal{A}_{1,1})}{\sqrt{2}},\qquad\tilde{\mathcal{A}}_{1,1}=\frac{(\mathcal{A}_{1,0}+\mathcal{A}_{1,1})}{\sqrt{2}}.
\end{equation}
The classical bound $\beta_C$, of the above Bell inequalities $\mathcal{I}_l$ \eqref{BE1Nm} for any $l$ is upper bounded as $\beta_C\leqslant (\sqrt{2}+1)(N-1)$. The following observables of the external parties
\begin{eqnarray}\label{GHZObsm}
\mathcal{A}'_{1,0}&=&\frac{X+Z}{\sqrt{2}},\quad \mathcal{A}'_{1,1}= \frac{X-Z}{\sqrt{2}},
\quad \mathcal{A}'_{1,2}=Y, \nonumber\\\mathcal{A}'_{i,0}&=&Z,\quad \mathcal{A}'_{i,1}=X,\quad \mathcal{A}'_{i,2}=Y\quad (i=2,\ldots,N)
\end{eqnarray}
as well as the GHZ-like states
\begin{equation}\label{GHZvecsm}
\ket{\phi_l}=\frac{1}{\sqrt{2}}(\ket{l_1\ldots l_N}+(-1)^{l_{1}}|\overline{l}_1\ldots\overline{l}_N\rangle),
\end{equation}
where $l_1,l_2,\ldots,l_N=0,1$ and $\overline{l}_i=1-l_i$ attain the value $\beta_Q=3(N-1)$ of \eqref{BE1Nm} for any $l$. This is also the maximum value attainable in quantum theory, also referred to as the quantum bound. Notice that, the measurements referred to in Eq. \eqref{GHZObsm} for every party form a two-dimensional tomographically complete set of measurements.

Suppose now that $L$ obtains an outcome $l$ when $e=0$. We impose that the corresponding post-measured state at $A_i's$ attains the quantum bound of $\langle\mathcal{I}_l\rangle$ \eqref{BE1Nm}, that is, we impose $\langle\mathcal{I}_l\otimes M_l\rangle_{\ket{\psi_{A\overline{A}}}}=p(l)\beta_Q$ where $p(l)$ is the probability of obtaining outcome $l$ at $L$. This allows us to self-test the post-measured states and the observables of $A_i's$ to be the reference ones \eqref{GHZvecsm} and \eqref{GHZObsm} according to Eqs. \eqref{ststate} and \eqref{stmea} respectively. Considering that local statistics of $L's$ measurement $p(l)=1/2^N$, one can then self-test the state generated by the sources to be maximally entangled (Eq. \eqref{GHZvecsm} for $N=2, l_i=0$) according to \eqref{stmea}. Moreover, $L's$ measurement is certified to be GHZ-like, that is,  $M_l=U_L\proj{\phi_l}U_L^{\dagger}$ for all $l$ according to \eqref{stmea}. In the appendix, we present the self-testing statement in full mathematical generality with the proof deferred to \cite{sarkar2024universal}. Let us now proceed to the second step of the scheme.

{\it{Step 2.}} Consider a reference unitary $\mathcal{U}:(\mathbb{C}^{2})^{\otimes N}\rightarrow (\mathbb{C}^{2})^{\otimes N}$ such that $\mathcal{U}\ket{\phi_l}=\ket{\delta_l}$ for all $l=0,1,\ldots,2^N-1$.
Here $\ket{\phi_l}$ are the GHZ-like vectors stated in Eq. \eqref{GHZvecsm} and $\{\ket{\delta_l}\}$ are an arbitrary orthonormal set of vectors that span the Hilbert space $(\mathbb{C}^{2})^{\otimes N}$. Consequently, any unitary that maps vectors from $(\mathbb{C}^{2})^{\otimes N}$ to $(\mathbb{C}^{2})^{\otimes N}$ can be represented as
\begin{equation}\label{genU}
\mathcal{U}=\sum_{l=1}^{2^N}\ket{\phi_l}\!\bra{\delta_l}.
\end{equation}

\begin{figure*}[t!]
    \centering
    \includegraphics[width=\linewidth]{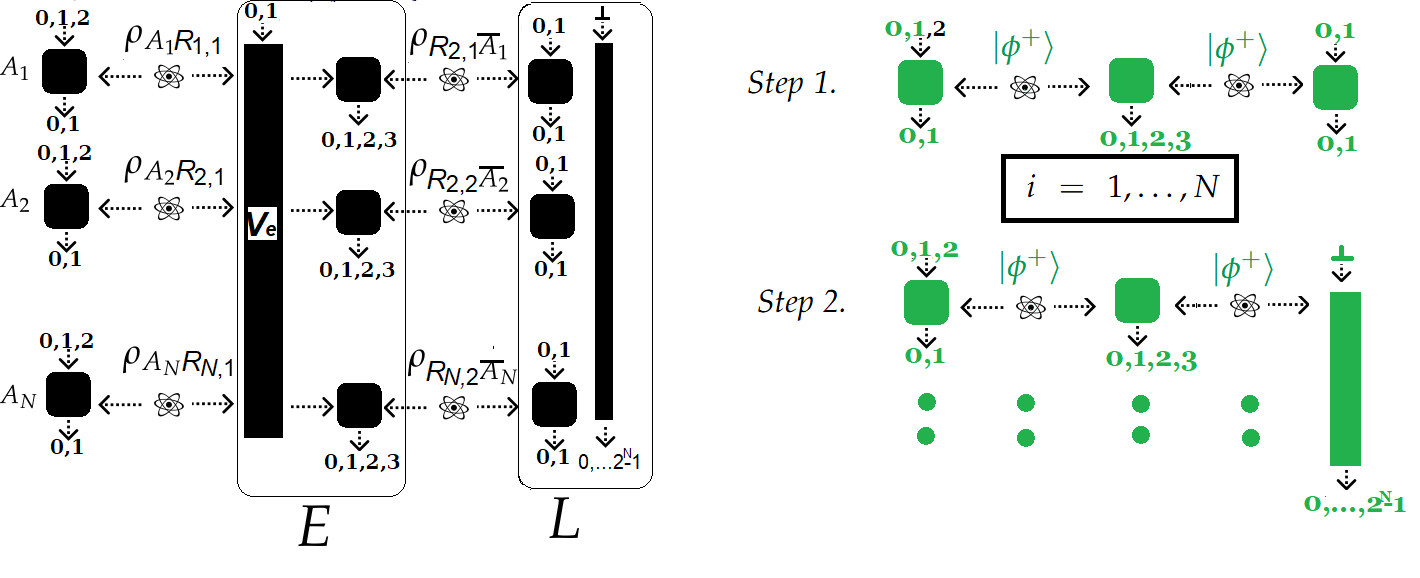}
    \caption{\textbf{DI certification of any quantum gate.} (Left.) Scenario. The setup involves \( N+2 \) parties: \( A_i \) for \( i = 1, \ldots, N \), a central party \( E \) and $L$. There are \( 2 N \) independent sources, with $N$ of them distributing a bipartite quantum state \( \rho_{A_iR_{i,1}} \) between party \( A_i \) and \( E \) and \(\rho_{R_{i,2}\overline{A}_i} \) between party \( E \) and $L$. Each \( A_i \) has three possible inputs and two outcomes. The central party \( E \) has two inputs that correspond to two different transformations $V_e$ such that $V_0=\I, V_1=V$. The transformed states are then fed to a quantum repeater that performs a single joint measurement on the $i-th$ subsystem along with its counterpart received from the other source. $L$ performs $2^N+1$ measurements which comprises of $N$ binary measurements and a single \( 2^N \)-outcome measurement. (Right.) The first two steps when $(e=0)$ of the DI certification are visualised. In {\it{Step 1.}}, each subnetwork is certified, which involves certifying the states, quantum repeater and two binary measurements of $A_i,L$. In {\it{Step 2.}}, using the characterised parts of the network from {\it{Step 1.}}, the rest of the components are certified. In the final Step ({\it{Step 3.}} see below) $E$ chooses $e=1$, which, using the characterised state and measurements from {\it{Step 2.}}, allows us to self-test the unitary $V$.  
}
    \label{fig2}
\end{figure*}

Taking $N$ to be arbitrarily large, we obtain any possible unitary transformation that maps vectors from any finite-dimensional Hilbert space to itself (up to isomorphism). Moreover, any projector $\proj{\delta_l}$ acting on $(\mathbb{C}^{2})^{\otimes N}$ can be expressed as
\begin{equation}\label{proj1m}
\proj{\delta_{l}}=\sum_{i_1,\ldots,i_N=0,1,2,3}f_{l,i_1,\ldots,i_N}\bigotimes_{k=1}^N\tilde{S}_{k,i_k}
\end{equation}
where $\tilde{S}_{k,0}=Z$, $\quad\tilde{S}_{k,1}=X$, $\tilde{S}_{k,2}=Y$ and $\tilde{S}_{k,3}=\I_2$ for any $k$ and every $f_{l,i_1,\ldots,i_N}$ is real for any $i_1,\ldots,i_N$.


Now, we consider the correlations $\{p(a_1\ldots a_Nl|x_1\ldots x_N 1)\}$ corresponding to $E's$ choice $e=1$. Using it, all the parties observe the following condition:
\begin{equation}\label{betstatfullm}
\sum_{i_1,\ldots,i_N=0}^3f_{l,i_1,\ldots,i_N}\left\langle\tilde{\mathcal{A}}_{1,i_1}\bigotimes_{k=2}^N \mathcal{A}_{k,i_k}\otimes V^{\dagger}M_lV\right\rangle_{\ket{\psi_{A\overline{A}}}} =\frac{1}{2^N}
\end{equation}
where $\tilde{A}_{1,0}, \tilde{A}_{1,1}$ are given in \eqref{overAm} with $\tilde{A}_{1,2}=A_{1,2}, \tilde{A}_{1,3}=\I$ and $A_{k,3}=\I$ for any $k$. 
Notice that the above statistics can be realised if the sources generate the two-qubit maximally entangled state and with $M_l=\{\proj{\phi_l}\}$ and $V=\mathcal{U}^*$ with $\mathcal{A}_{i,j}=\mathcal{A}'_{i,j}$ (Eq. \eqref{GHZObsm}). Let us now state the main result of this work, recalling that $V:\mathcal{H}_L\rightarrow\mathcal{H}_L$ where $\mathcal{H}_L=\bigotimes_{i=1}^N\mathcal{H}_{\overline{A}_i}$.

\begin{thm}\label{thm1}
    Suppose that the states generated by the sources and the measurements of $A_i's$ and $L$ are certified as suggested in {\it{Step 1}} above when $e=0$. Consider then that when $e=1$ and all $A_i's$ and $L$ observe the condition \eqref{betstatfullm}. Consequently, the transformation $V$ is certified to be $\mathcal{U}$ \eqref{genU} according to the definition \eqref{stU}.
\end{thm}

The proof of the above theorem is provide in Appendix. Let us now proceed towards DI characterisation of $V$.

{\it{DI certification of any quantum gate---}} Consider now that the unitary $V$ can change the support of the local states, that is, $V:\mathcal{H}_L\rightarrow \mathcal{H}'_L$ where $\mathcal{H}'_L$ is isomorphic to $\mathcal{H}_L$. In the almost DI scheme [Fig. \ref{fig1}], we utilise the fact that $L's$ measurement is certified on $\mathcal{H}_L$ when $e=0$ and thus when $e=1$, $V$ should not alter $\mathcal{H}_L$ and can be certified using it. Now, if $V$ changes $\mathcal{H}_L$ to $\mathcal{H}'_L$, then we can not use the measurement at $L$, as it is not certified on the space $\mathcal{H}'_L$. 

To avoid this problem, we utilise the idea from \cite{Allst, sarkar2024steer}, which ensures that the support of the local states at $L$ remains invariant. The basic premise of the idea is to teleport the states from $E$ to $L$ rather than directly sending them. Consequently, we modify the setup from Fig. \ref{fig1} as shown in Fig. \ref{fig2} by introducing more components at the central party $E$. Here, each of the $N$ subsystems from the output of $V_e$ are fed to a quantum repeater (a four-outcome Bell-measurement) which performs a joint measurement on them with additional $N$ subsystems generated by bipartite sources $P_i's$ [see Fig. \ref{fig2}]. The outcomes of all these measurements are denoted by $r_i=0,1,2,3$. 
Moreover, the additional $N$ subsystems generated by $P_i's$ go to $L$. Moreover, at $L$ additional to the single joint measurement on all subsystems, $N$ distinct measurement boxes are taken that perform two measurements each on the $N$ different subsystems. The inputs of these boxes are denoted as $y\equiv y_1\ldots y_N$ with outputs denoted by $l\equiv l_1\ldots l_N$ where $y_i,l_i=0,1$ are the input and output of the $i-th$ box. Consequently, $L$ now performs $2^N+1$ measurements with the input of the single joint measurement on all $N$ subsystem denoted as $y=\perp$ and outcomes denoted as $l\equiv l_1\ldots l_N$.
The rest of the setup is the same as Fig. \ref{fig1}. 

All the parties repeat the experiment to construct the correlations $\vec{p}=\{p(a_1\ldots a_Nr_1\ldots r_Nl|x_1\ldots x_N e y)\}$. Here, we denote the states generated by $P_i$ and $P_i'$ as $\rho_{A_iR_i}$ and $\rho_{R_i\overline{A}_i}$ and thus all the states together are referred to as $\rho_{AR}=\bigotimes_{i=1}^N\rho_{A_iR_i}$ and $\rho_{R\overline{A}}=\bigotimes_{i=1}^N\rho_{R_i\overline{A}_i}$. As done before, due to no restriction on the dimension of the Hilbert spaces, we take these states to be pure. 

The self-testing scheme is divided into three parts: In the first part, Eve chooses $e=0$, which is used to certify the states generated by the sources $P_i,P_i'$ and the measurements at $E$. The second and the third step are identical to the almost-DI scheme. Let us begin by describing the first step of the DI scheme to certify any unitary transformation.


{\it{Step 1.}} Consider the correlations $p(a_ir_il_i|x_i,e=0,y_i)$ for any $i=1,\ldots,N$. Using these correlations, one can construct Bell inequalities for each $r_i$ as 
$\langle \mathcal{K}_{r_i,i} \rangle\leq \beta_C,\ (\beta_C=\sqrt{2})$ where 
\begin{eqnarray}\label{K1}
    \mathcal{K}_{r_1,1}= (-1)^{r_{i,1}}\tilde{\mathcal{A}}_{1,1} \mathcal{B}_{1,1} +(-1)^{r_{i,2}}\tilde{\mathcal{A}}_{1,0}\mathcal{B}_{1,0} 
\end{eqnarray}
where $ \tilde{\mathcal{A}}_{1,0}=\frac{(\mathcal{A}_{1,0}-\mathcal{A}_{1,1})}{\sqrt{2}},\ \tilde{\mathcal{A}}_{1,1}=\frac{(\mathcal{A}_{1,0}+\mathcal{A}_{1,1})}{\sqrt{2}}$ and for $i=2,\ldots,N$
\begin{eqnarray}\label{K2}
    \mathcal{K}_{r_i,i}= (-1)^{r_{i,1}}\mathcal{A}_{i,1} \tilde{\mathcal{B}}_{i,1} +(-1)^{r_{i,2}}\mathcal{A}_{i,0}\tilde{\mathcal{B}}_{i,0} 
\end{eqnarray}
where $ \tilde{\mathcal{B}}_{i,0}=\frac{(\mathcal{B}_{i,0}-\mathcal{B}_{i,1})}{\sqrt{2}},\ \tilde{\mathcal{B}}_{1,1}=\frac{(\mathcal{B}_{i,0}+\mathcal{B}_{i,1})}{\sqrt{2}}$. Here, $\mathcal{B}_{i,j}$ are the observables of the measurements at $L$ defined similarly as $\mathcal{A}_{i,j}$. Notice that they are identical to the Bell functionals $\mathcal{I}_l$ \eqref{BE1Nm} for $N=2$ (removing the third measurement). The quantum bound of the $\langle \mathcal{K}_{r_i,i} \rangle$ is $2$. Consequently, when $\langle \mathcal{K}_{r_i,i}\rangle=2$ for all $r_i,i$, then the observables $\mathcal{A}_{i,j}$ for $j=1,2$ are certified to be the reference ones \eqref{GHZObsm} according to Eq. \eqref{stmea} along with the observables $\mathcal{B}_{i,j}$. Importantly, the state generated by the sources $P_i,P_i'$ for every $i$ are certified to be maximally entangled according to Eq. \eqref{ststate} and the measurement (quantum repeater) at $E$ is certified to be the Bell-basis measurement according to Eq. \eqref{stmea}.

{\it{Step 2.}} Now, we again repeat {\it{Step 1.}} of the almost DI scheme by considering $\{p(a_1\ldots a_N0\ldots 0l|x_1\ldots x_N e=0 y=\perp)\}$. When $L$ observes the outcome $l$, then $A_i's$ observe Bell inequalities $\mathcal{I}_l$ Eq. \eqref{BE1Nm} are maximally violated, that is, we impose $\langle\mathcal{I}_l\otimes R_0\otimes M_l\rangle_{\rho_{AR}\otimes\rho_{R\overline{A}}}=p(l)\overline{p}(0)\beta_Q$ where $R_0\equiv \bigotimes_{i=1}^N R_{i,0}$ and $\overline{p}(0)=\prod_{i=1}^Np(r_i=0)$. Consequently, using the certified states and measurements from {\it{Step 1.}}, we arrive at the condition $\langle\mathcal{I}_l\otimes M_l\rangle_{\ket{\tau_{A\overline{A}}}}=p(l)\beta_Q$ which is the same as in {\it{Step 1.}} of the almost DI scheme with the state $\ket{\tau_{A\overline{A}}}$ already certified to be maximally entangled according to Eq. \eqref{ststate}. This allows us to certify the post-measured state of $L$ and $E$ and the left-over observable $\mathcal{A}_{i,2}$. Then, using the fact that each outcome of $L$ is equi-probable allows us to certify $L's$ measurement corresponding to input $l=\perp$ be GHZ-like according to Eq. \eqref{stmea}. 

{\it{Step 3.}} Then, {\it{Step 2.}} of the almost DI scheme is applied by considering $\{p(a_1\ldots a_N0\ldots 0l|x_1\ldots x_N e=1)\}$. Here, we thus consider the correlations among all the parties similar to Eq. \eqref{betstatfullm} given by
\begin{eqnarray}\label{betstatfullm12121}
\sum_{i_1,\ldots,i_N=0}^3f_{l,i_1,\ldots,i_N}\left\langle\tilde{\mathcal{A}}_{1,i_1}\bigotimes_{k=2}^N \mathcal{A}_{k,i_k}\otimes V^{\dagger}R_0V\otimes M_l\right\rangle_{\rho_{AR}\otimes\rho_{R\overline{A}}} \nonumber\\=\frac{\overline{p}(0)}{2^N}.\qquad
\end{eqnarray}
The above condition along with the certified state and measurements from {\it{Step 1.}} and {\it{Step 2.}} allows us to estabilish the final result of this work.

\begin{thm}\label{theorem4m}
Suppose that the states generated by the sources and the measurements of $A_i's,E$ and $L$ are certified as suggested in {\it{Step 1}} and {\it{Step 2}} above when $e=0$. Consider then that when $e=1$ and all $A_i's$ and $L$ observe the condition \eqref{betstatfullm12121}. Consequently, the transformation $V$ projected onto the support of the local state at $E$ is certified to be $\mathcal{U}$ \eqref{genU} according to Eq. \eqref{stU}.
\end{thm}
In the Appendix, we give a formal statement along with the proof of the above stated theorem.
Thus, we have now provided a DI way of certifying any quantum unitary.

{\it{Outlook---}} Following our work, we have identified several open problems that might interesting from theoretical as well as practical perspective. Although we have provided a proof-of-principle scheme for self-testing arbitrary unitary transformations, practical implementations will inevitably encounter deviations from the ideal statistics. A key open question is whether our scheme remains robust under such experimental imperfections. Furthermore, the current construction relies on highly entangled states and complex measurements, which may be resource-intensive to realise. It would therefore be valuable to examine whether similar self-testing guarantees can be achieved using partially entangled states or simpler measurement settings. Lastly, it would be interesting if our scenario can be utilised for self-testing general quantum processes that are completely positive and trace-prseserving.

\begin{acknowledgments}
 This project was funded within the QuantERA II Programme (VERIqTAS project) that has received funding from the European Union’s Horizon 2020 research and innovation programme under Grant Agreement No 101017733 and from the Polish National Science Center (project No 2021/03/Y/ST2/00175). We also acknowledge National Science Centre, Poland, grant Opus 25, 2023/49/B/ST2/02468.
\end{acknowledgments}


\begin{thebibliography}{30}%
\makeatletter
\providecommand \@ifxundefined [1]{%
 \@ifx{#1\undefined}
}%
\providecommand \@ifnum [1]{%
 \ifnum #1\expandafter \@firstoftwo
 \else \expandafter \@secondoftwo
 \fi
}%
\providecommand \@ifx [1]{%
 \ifx #1\expandafter \@firstoftwo
 \else \expandafter \@secondoftwo
 \fi
}%
\providecommand \natexlab [1]{#1}%
\providecommand \enquote  [1]{``#1''}%
\providecommand \bibnamefont  [1]{#1}%
\providecommand \bibfnamefont [1]{#1}%
\providecommand \citenamefont [1]{#1}%
\providecommand \href@noop [0]{\@secondoftwo}%
\providecommand \href [0]{\begingroup \@sanitize@url \@href}%
\providecommand \@href[1]{\@@startlink{#1}\@@href}%
\providecommand \@@href[1]{\endgroup#1\@@endlink}%
\providecommand \@sanitize@url [0]{\catcode `\\12\catcode `\$12\catcode `\&12\catcode `\#12\catcode `\^12\catcode `\_12\catcode `\%12\relax}%
\providecommand \@@startlink[1]{}%
\providecommand \@@endlink[0]{}%
\providecommand \url  [0]{\begingroup\@sanitize@url \@url }%
\providecommand \@url [1]{\endgroup\@href {#1}{\urlprefix }}%
\providecommand \urlprefix  [0]{URL }%
\providecommand \Eprint [0]{\href }%
\providecommand \doibase [0]{https://doi.org/}%
\providecommand \selectlanguage [0]{\@gobble}%
\providecommand \bibinfo  [0]{\@secondoftwo}%
\providecommand \bibfield  [0]{\@secondoftwo}%
\providecommand \translation [1]{[#1]}%
\providecommand \BibitemOpen [0]{}%
\providecommand \bibitemStop [0]{}%
\providecommand \bibitemNoStop [0]{.\EOS\space}%
\providecommand \EOS [0]{\spacefactor3000\relax}%
\providecommand \BibitemShut  [1]{\csname bibitem#1\endcsname}%
\let\auto@bib@innerbib\@empty
\bibitem [{\citenamefont {Bell}(1964)}]{Bell}%
  \BibitemOpen
  \bibfield  {author} {\bibinfo {author} {\bibfnamefont {J.~S.}\ \bibnamefont {Bell}},\ }\bibfield  {title} {\bibinfo {title} {On the {E}instein-{P}odolsky-{R}osen paradox},\ }\href {https://doi.org/10.1103/PhysicsPhysiqueFizika.1.195} {\bibfield  {journal} {\bibinfo  {journal} {Physics Physique Fizika}\ }\textbf {\bibinfo {volume} {1}},\ \bibinfo {pages} {195} (\bibinfo {year} {1964})}\BibitemShut {NoStop}%
\bibitem [{\citenamefont {Bell}(1966)}]{Bell66}%
  \BibitemOpen
  \bibfield  {author} {\bibinfo {author} {\bibfnamefont {J.~S.}\ \bibnamefont {Bell}},\ }\bibfield  {title} {\bibinfo {title} {On the problem of hidden variables in quantum mechanics},\ }\href {https://doi.org/10.1103/RevModPhys.38.447} {\bibfield  {journal} {\bibinfo  {journal} {Rev. Mod. Phys.}\ }\textbf {\bibinfo {volume} {38}},\ \bibinfo {pages} {447} (\bibinfo {year} {1966})}\BibitemShut {NoStop}%
\bibitem [{\citenamefont {Mayers}\ and\ \citenamefont {Yao}(1998)}]{Mayers_selftesting}%
  \BibitemOpen
  \bibfield  {author} {\bibinfo {author} {\bibfnamefont {D.}~\bibnamefont {Mayers}}\ and\ \bibinfo {author} {\bibfnamefont {A.}~\bibnamefont {Yao}},\ }\bibfield  {title} {\bibinfo {title} {Quantum cryptography with imperfect apparatus},\ }in\ \href {https://doi.org/10.1109/SFCS.1998.743501} {\emph {\bibinfo {booktitle} {Proceedings 39th Annual Symposium on Foundations of Computer Science (Cat. No. 98CB36280)}}}\ (\bibinfo {organization} {IEEE},\ \bibinfo {year} {1998})\ pp.\ \bibinfo {pages} {503--509}\BibitemShut {NoStop}%
\bibitem [{\citenamefont {Mayers}\ and\ \citenamefont {Yao}(2004)}]{Yao}%
  \BibitemOpen
  \bibfield  {author} {\bibinfo {author} {\bibfnamefont {D.}~\bibnamefont {Mayers}}\ and\ \bibinfo {author} {\bibfnamefont {A.}~\bibnamefont {Yao}},\ }\bibfield  {title} {\bibinfo {title} {Self testing quantum apparatus},\ }\href {https://doi.org/doi.org/10.26421/QIC4.4} {\bibfield  {journal} {\bibinfo  {journal} {Quantum Inf. Comput.}\ }\textbf {\bibinfo {volume} {4}},\ \bibinfo {pages} {273} (\bibinfo {year} {2004})}\BibitemShut {NoStop}%
\bibitem [{\citenamefont {McKague}\ \emph {et~al.}(2012)\citenamefont {McKague}, \citenamefont {Yang},\ and\ \citenamefont {Scarani}}]{Scarani}%
  \BibitemOpen
  \bibfield  {author} {\bibinfo {author} {\bibfnamefont {M.}~\bibnamefont {McKague}}, \bibinfo {author} {\bibfnamefont {T.~H.}\ \bibnamefont {Yang}},\ and\ \bibinfo {author} {\bibfnamefont {V.}~\bibnamefont {Scarani}},\ }\bibfield  {title} {\bibinfo {title} {Robust self-testing of the singlet},\ }\href {https://doi.org/10.1088/1751-8113/45/45/455304} {\bibfield  {journal} {\bibinfo  {journal} {J. Phys. A: Math. Theor.}\ }\textbf {\bibinfo {volume} {45}},\ \bibinfo {pages} {455304} (\bibinfo {year} {2012})}\BibitemShut {NoStop}%
\bibitem [{\citenamefont {Reichardt}\ \emph {et~al.}(2013{\natexlab{a}})\citenamefont {Reichardt}, \citenamefont {Unger},\ and\ \citenamefont {Vazirani}}]{Reichardt_nature}%
  \BibitemOpen
  \bibfield  {author} {\bibinfo {author} {\bibfnamefont {B.}~\bibnamefont {Reichardt}}, \bibinfo {author} {\bibfnamefont {F.}~\bibnamefont {Unger}},\ and\ \bibinfo {author} {\bibfnamefont {U.}~\bibnamefont {Vazirani}},\ }\bibfield  {title} {\bibinfo {title} {Classical command of quantum systems},\ }\href {https://doi.org/10.1038/nature12035} {\bibfield  {journal} {\bibinfo  {journal} {Nature}\ }\textbf {\bibinfo {volume} {496}},\ \bibinfo {pages} {456} (\bibinfo {year} {2013}{\natexlab{a}})}\BibitemShut {NoStop}%
\bibitem [{\citenamefont {McKague}(2014)}]{Mckague_2014}%
  \BibitemOpen
  \bibfield  {author} {\bibinfo {author} {\bibfnamefont {M.}~\bibnamefont {McKague}},\ }\bibfield  {title} {\bibinfo {title} {Self-testing graph states},\ }in\ \href {https://doi.org/10.1007/978-3-642-54429-3_7} {\emph {\bibinfo {booktitle} {Theory of Quantum Computation, Communication, and Cryptography}}},\ \bibinfo {editor} {edited by\ \bibinfo {editor} {\bibfnamefont {D.}~\bibnamefont {Bacon}}, \bibinfo {editor} {\bibfnamefont {M.}~\bibnamefont {Martin-Delgado}},\ and\ \bibinfo {editor} {\bibfnamefont {M.}~\bibnamefont {Roetteler}}}\ (\bibinfo  {publisher} {Springer Berlin Heidelberg},\ \bibinfo {address} {Berlin, Heidelberg},\ \bibinfo {year} {2014})\ pp.\ \bibinfo {pages} {104--120}\BibitemShut {NoStop}%
\bibitem [{\citenamefont {Wu}\ \emph {et~al.}(2014)\citenamefont {Wu}, \citenamefont {Cai}, \citenamefont {Yang}, \citenamefont {Le}, \citenamefont {Bancal},\ and\ \citenamefont {Scarani}}]{Wu_2014}%
  \BibitemOpen
  \bibfield  {author} {\bibinfo {author} {\bibfnamefont {X.}~\bibnamefont {Wu}}, \bibinfo {author} {\bibfnamefont {Y.}~\bibnamefont {Cai}}, \bibinfo {author} {\bibfnamefont {T.~H.}\ \bibnamefont {Yang}}, \bibinfo {author} {\bibfnamefont {H.~N.}\ \bibnamefont {Le}}, \bibinfo {author} {\bibfnamefont {J.-D.}\ \bibnamefont {Bancal}},\ and\ \bibinfo {author} {\bibfnamefont {V.}~\bibnamefont {Scarani}},\ }\bibfield  {title} {\bibinfo {title} {Robust self-testing of the three-qubit {$W$} state},\ }\href {https://doi.org/10.1103/PhysRevA.90.042339} {\bibfield  {journal} {\bibinfo  {journal} {Phys. Rev. A}\ }\textbf {\bibinfo {volume} {90}},\ \bibinfo {pages} {042339} (\bibinfo {year} {2014})}\BibitemShut {NoStop}%
\bibitem [{\citenamefont {Bamps}\ and\ \citenamefont {Pironio}(2015)}]{Bamps}%
  \BibitemOpen
  \bibfield  {author} {\bibinfo {author} {\bibfnamefont {C.}~\bibnamefont {Bamps}}\ and\ \bibinfo {author} {\bibfnamefont {S.}~\bibnamefont {Pironio}},\ }\bibfield  {title} {\bibinfo {title} {Sum-of-squares decompositions for a family of {C}lauser-{H}orne-{S}himony-{H}olt-like inequalities and their application to self-testing},\ }\href {https://doi.org/10.1103/PhysRevA.91.052111} {\bibfield  {journal} {\bibinfo  {journal} {Phys. Rev. A}\ }\textbf {\bibinfo {volume} {91}},\ \bibinfo {pages} {052111} (\bibinfo {year} {2015})}\BibitemShut {NoStop}%
\bibitem [{\citenamefont {Wang}\ \emph {et~al.}(2016)\citenamefont {Wang}, \citenamefont {Wu},\ and\ \citenamefont {Scarani}}]{All}%
  \BibitemOpen
  \bibfield  {author} {\bibinfo {author} {\bibfnamefont {Y.}~\bibnamefont {Wang}}, \bibinfo {author} {\bibfnamefont {X.}~\bibnamefont {Wu}},\ and\ \bibinfo {author} {\bibfnamefont {V.}~\bibnamefont {Scarani}},\ }\bibfield  {title} {\bibinfo {title} {All the self-testings of the singlet for two binary measurements},\ }\href {https://doi.org/10.1088/1367-2630/18/2/025021} {\bibfield  {journal} {\bibinfo  {journal} {New J. Phys.}\ }\textbf {\bibinfo {volume} {18}},\ \bibinfo {pages} {025021} (\bibinfo {year} {2016})}\BibitemShut {NoStop}%
\bibitem [{\citenamefont {{\v{S}}upi{\'{c}}}\ \emph {et~al.}(2016)\citenamefont {{\v{S}}upi{\'{c}}}, \citenamefont {Augusiak}, \citenamefont {Salavrakos},\ and\ \citenamefont {Ac{\'{\i}}n}}]{chainedBell}%
  \BibitemOpen
  \bibfield  {author} {\bibinfo {author} {\bibfnamefont {I.}~\bibnamefont {{\v{S}}upi{\'{c}}}}, \bibinfo {author} {\bibfnamefont {R.}~\bibnamefont {Augusiak}}, \bibinfo {author} {\bibfnamefont {A.}~\bibnamefont {Salavrakos}},\ and\ \bibinfo {author} {\bibfnamefont {A.}~\bibnamefont {Ac{\'{\i}}n}},\ }\bibfield  {title} {\bibinfo {title} {Self-testing protocols based on the chained {B}ell inequalities},\ }\href {https://doi.org/10.1088/1367-2630/18/3/035013} {\bibfield  {journal} {\bibinfo  {journal} {New J. Phys.}\ }\textbf {\bibinfo {volume} {18}},\ \bibinfo {pages} {035013} (\bibinfo {year} {2016})}\BibitemShut {NoStop}%
\bibitem [{\citenamefont {Coladangelo}\ \emph {et~al.}(2017)\citenamefont {Coladangelo}, \citenamefont {Goh},\ and\ \citenamefont {Scarani}}]{Projection}%
  \BibitemOpen
  \bibfield  {author} {\bibinfo {author} {\bibfnamefont {A.}~\bibnamefont {Coladangelo}}, \bibinfo {author} {\bibfnamefont {K.~T.}\ \bibnamefont {Goh}},\ and\ \bibinfo {author} {\bibfnamefont {V.}~\bibnamefont {Scarani}},\ }\bibfield  {title} {\bibinfo {title} {All pure bipartite entangled states can be self-tested},\ }\href {https://doi.org/10.1038/ncomms15485} {\bibfield  {journal} {\bibinfo  {journal} {Nature Communications}\ }\textbf {\bibinfo {volume} {8}},\ \bibinfo {pages} {15485} (\bibinfo {year} {2017})}\BibitemShut {NoStop}%
\bibitem [{\citenamefont {Kaniewski}\ \emph {et~al.}(2019)\citenamefont {Kaniewski}, \citenamefont {{\v{S}}upi{\'{c}}}, \citenamefont {Tura}, \citenamefont {Baccari}, \citenamefont {Salavrakos},\ and\ \citenamefont {Augusiak}}]{Jed1}%
  \BibitemOpen
  \bibfield  {author} {\bibinfo {author} {\bibfnamefont {J.}~\bibnamefont {Kaniewski}}, \bibinfo {author} {\bibfnamefont {I.}~\bibnamefont {{\v{S}}upi{\'{c}}}}, \bibinfo {author} {\bibfnamefont {J.}~\bibnamefont {Tura}}, \bibinfo {author} {\bibfnamefont {F.}~\bibnamefont {Baccari}}, \bibinfo {author} {\bibfnamefont {A.}~\bibnamefont {Salavrakos}},\ and\ \bibinfo {author} {\bibfnamefont {R.}~\bibnamefont {Augusiak}},\ }\bibfield  {title} {\bibinfo {title} {Maximal nonlocality from maximal entanglement and mutually unbiased bases, and self-testing of two-qutrit quantum systems},\ }\href {https://doi.org/10.22331/q-2019-10-24-198} {\bibfield  {journal} {\bibinfo  {journal} {{Quantum}}\ }\textbf {\bibinfo {volume} {3}},\ \bibinfo {pages} {198} (\bibinfo {year} {2019})}\BibitemShut {NoStop}%
\bibitem [{\citenamefont {Sarkar}\ \emph {et~al.}(2021)\citenamefont {Sarkar}, \citenamefont {Saha}, \citenamefont {Kaniewski},\ and\ \citenamefont {Augusiak}}]{sarkar}%
  \BibitemOpen
  \bibfield  {author} {\bibinfo {author} {\bibfnamefont {S.}~\bibnamefont {Sarkar}}, \bibinfo {author} {\bibfnamefont {D.}~\bibnamefont {Saha}}, \bibinfo {author} {\bibfnamefont {J.}~\bibnamefont {Kaniewski}},\ and\ \bibinfo {author} {\bibfnamefont {R.}~\bibnamefont {Augusiak}},\ }\bibfield  {title} {\bibinfo {title} {Self-testing quantum systems of arbitrary local dimension with minimal number of measurements},\ }\href {https://doi.org/10.1038/s41534-021-00490-3} {\bibfield  {journal} {\bibinfo  {journal} {npj Quantum Information}\ }\textbf {\bibinfo {volume} {7}},\ \bibinfo {pages} {151} (\bibinfo {year} {2021})}\BibitemShut {NoStop}%
\bibitem [{\citenamefont {Sarkar}\ and\ \citenamefont {Augusiak}(2022)}]{sarkaro2}%
  \BibitemOpen
  \bibfield  {author} {\bibinfo {author} {\bibfnamefont {S.}~\bibnamefont {Sarkar}}\ and\ \bibinfo {author} {\bibfnamefont {R.}~\bibnamefont {Augusiak}},\ }\bibfield  {title} {\bibinfo {title} {Self-testing of multipartite greenberger-horne-zeilinger states of arbitrary local dimension with arbitrary number of measurements per party},\ }\href {https://doi.org/10.1103/PhysRevA.105.032416} {\bibfield  {journal} {\bibinfo  {journal} {Phys. Rev. A}\ }\textbf {\bibinfo {volume} {105}},\ \bibinfo {pages} {032416} (\bibinfo {year} {2022})}\BibitemShut {NoStop}%
\bibitem [{\citenamefont {{\v{S}}upi{\'c}}\ \emph {et~al.}(2023)\citenamefont {{\v{S}}upi{\'c}}, \citenamefont {Bowles}, \citenamefont {Renou}, \citenamefont {Ac{\'\i}n},\ and\ \citenamefont {Hoban}}]{Allst}%
  \BibitemOpen
  \bibfield  {author} {\bibinfo {author} {\bibfnamefont {I.}~\bibnamefont {{\v{S}}upi{\'c}}}, \bibinfo {author} {\bibfnamefont {J.}~\bibnamefont {Bowles}}, \bibinfo {author} {\bibfnamefont {M.-O.}\ \bibnamefont {Renou}}, \bibinfo {author} {\bibfnamefont {A.}~\bibnamefont {Ac{\'\i}n}},\ and\ \bibinfo {author} {\bibfnamefont {M.~J.}\ \bibnamefont {Hoban}},\ }\bibfield  {title} {\bibinfo {title} {Quantum networks self-test all entangled states},\ }\href {https://doi.org/10.1038/s41567-023-01945-4} {\bibfield  {journal} {\bibinfo  {journal} {Nat. Phys.}\ }\textbf {\bibinfo {volume} {19}},\ \bibinfo {pages} {670} (\bibinfo {year} {2023})}\BibitemShut {NoStop}%
\bibitem [{\citenamefont {Sarkar}\ \emph {et~al.}(2025)\citenamefont {Sarkar}, \citenamefont {Datta}, \citenamefont {Halder},\ and\ \citenamefont {Augusiak}}]{sarkarPRL}%
  \BibitemOpen
  \bibfield  {author} {\bibinfo {author} {\bibfnamefont {S.}~\bibnamefont {Sarkar}}, \bibinfo {author} {\bibfnamefont {C.}~\bibnamefont {Datta}}, \bibinfo {author} {\bibfnamefont {S.}~\bibnamefont {Halder}},\ and\ \bibinfo {author} {\bibfnamefont {R.}~\bibnamefont {Augusiak}},\ }\bibfield  {title} {\bibinfo {title} {Self-testing composite measurements and bound entangled state in a single quantum network},\ }\href {https://doi.org/10.1103/PhysRevLett.134.190203} {\bibfield  {journal} {\bibinfo  {journal} {Phys. Rev. Lett.}\ }\textbf {\bibinfo {volume} {134}},\ \bibinfo {pages} {190203} (\bibinfo {year} {2025})}\BibitemShut {NoStop}%
\bibitem [{\citenamefont {Woodhead}\ \emph {et~al.}(2020)\citenamefont {Woodhead}, \citenamefont {Kaniewski}, \citenamefont {Bourdoncle}, \citenamefont {Salavrakos}, \citenamefont {Bowles}, \citenamefont {Ac\'in},\ and\ \citenamefont {Augusiak}}]{random1}%
  \BibitemOpen
  \bibfield  {author} {\bibinfo {author} {\bibfnamefont {E.}~\bibnamefont {Woodhead}}, \bibinfo {author} {\bibfnamefont {J.}~\bibnamefont {Kaniewski}}, \bibinfo {author} {\bibfnamefont {B.}~\bibnamefont {Bourdoncle}}, \bibinfo {author} {\bibfnamefont {A.}~\bibnamefont {Salavrakos}}, \bibinfo {author} {\bibfnamefont {J.}~\bibnamefont {Bowles}}, \bibinfo {author} {\bibfnamefont {A.}~\bibnamefont {Ac\'in}},\ and\ \bibinfo {author} {\bibfnamefont {R.}~\bibnamefont {Augusiak}},\ }\bibfield  {title} {\bibinfo {title} {Maximal randomness from partially entangled states},\ }\href {https://doi.org/https://doi.org/10.1103/PhysRevResearch.2.042028} {\bibfield  {journal} {\bibinfo  {journal} {Phys. Rev. Research}\ }\textbf {\bibinfo {volume} {2}},\ \bibinfo {pages} {042028} (\bibinfo {year} {2020})}\BibitemShut {NoStop}%
\bibitem [{\citenamefont {Chen}\ \emph {et~al.}(2023)\citenamefont {Chen}, \citenamefont {Volčič},\ and\ \citenamefont {Mančinska}}]{chen1}%
  \BibitemOpen
  \bibfield  {author} {\bibinfo {author} {\bibfnamefont {R.}~\bibnamefont {Chen}}, \bibinfo {author} {\bibfnamefont {J.}~\bibnamefont {Volčič}},\ and\ \bibinfo {author} {\bibfnamefont {L.}~\bibnamefont {Mančinska}},\ }\bibfield  {title} {\bibinfo {title} {All projective measurements can be self-tested},\ }\href {https://arxiv.org/abs/2302.00974} {\bibfield  {journal} {\bibinfo  {journal} {arXiv:2302.00974}\ } (\bibinfo {year} {2023})}\BibitemShut {NoStop}%
\bibitem [{\citenamefont {Renou}\ \emph {et~al.}(2018)\citenamefont {Renou}, \citenamefont {Kaniewski},\ and\ \citenamefont {Brunner}}]{Marco}%
  \BibitemOpen
  \bibfield  {author} {\bibinfo {author} {\bibfnamefont {M.-O.}\ \bibnamefont {Renou}}, \bibinfo {author} {\bibfnamefont {J.}~\bibnamefont {Kaniewski}},\ and\ \bibinfo {author} {\bibfnamefont {N.}~\bibnamefont {Brunner}},\ }\bibfield  {title} {\bibinfo {title} {Self-testing entangled measurements in quantum networks},\ }\href {https://doi.org/10.1103/PhysRevLett.121.250507} {\bibfield  {journal} {\bibinfo  {journal} {Phys. Rev. Lett.}\ }\textbf {\bibinfo {volume} {121}},\ \bibinfo {pages} {250507} (\bibinfo {year} {2018})}\BibitemShut {NoStop}%
\bibitem [{\citenamefont {Zhou}\ \emph {et~al.}(2022)\citenamefont {Zhou}, \citenamefont {Xu}, \citenamefont {Zhao}, \citenamefont {Zhen}, \citenamefont {Li}, \citenamefont {Liu},\ and\ \citenamefont {Chen}}]{JW2}%
  \BibitemOpen
  \bibfield  {author} {\bibinfo {author} {\bibfnamefont {Q.}~\bibnamefont {Zhou}}, \bibinfo {author} {\bibfnamefont {X.-Y.}\ \bibnamefont {Xu}}, \bibinfo {author} {\bibfnamefont {S.}~\bibnamefont {Zhao}}, \bibinfo {author} {\bibfnamefont {Y.-Z.}\ \bibnamefont {Zhen}}, \bibinfo {author} {\bibfnamefont {L.}~\bibnamefont {Li}}, \bibinfo {author} {\bibfnamefont {N.-L.}\ \bibnamefont {Liu}},\ and\ \bibinfo {author} {\bibfnamefont {K.}~\bibnamefont {Chen}},\ }\bibfield  {title} {\bibinfo {title} {Robust self-testing of multipartite {G}reenberger-{H}orne-{Z}eilinger-state measurements in quantum networks},\ }\href {https://doi.org/10.1103/PhysRevA.106.042608} {\bibfield  {journal} {\bibinfo  {journal} {Phys. Rev. A}\ }\textbf {\bibinfo {volume} {106}},\ \bibinfo {pages} {042608} (\bibinfo {year} {2022})}\BibitemShut {NoStop}%
\bibitem [{\citenamefont {Šupić}\ and\ \citenamefont {Brunner}(2022)}]{NLWEsupic}%
  \BibitemOpen
  \bibfield  {author} {\bibinfo {author} {\bibfnamefont {I.}~\bibnamefont {Šupić}}\ and\ \bibinfo {author} {\bibfnamefont {N.}~\bibnamefont {Brunner}},\ }\bibfield  {title} {\bibinfo {title} {Self-testing nonlocality without entanglement},\ }\href {https://arxiv.org/abs/2203.13171} {\bibfield  {journal} {\bibinfo  {journal} {arXiv:2203.13171}\ } (\bibinfo {year} {2022})}\BibitemShut {NoStop}%
\bibitem [{\citenamefont {Sarkar}\ \emph {et~al.}(2024{\natexlab{a}})\citenamefont {Sarkar}, \citenamefont {Alexandre C.~Orthey},\ and\ \citenamefont {Augusiak}}]{sarkar2024universal}%
  \BibitemOpen
  \bibfield  {author} {\bibinfo {author} {\bibfnamefont {S.}~\bibnamefont {Sarkar}}, \bibinfo {author} {\bibfnamefont {J.}~\bibnamefont {Alexandre C.~Orthey}},\ and\ \bibinfo {author} {\bibfnamefont {R.}~\bibnamefont {Augusiak}},\ }\href {https://arxiv.org/abs/2312.04405} {\bibinfo {title} {A universal scheme to self-test any quantum state and extremal measurement}} (\bibinfo {year} {2024}{\natexlab{a}}),\ \Eprint {https://arxiv.org/abs/2312.04405} {arXiv:2312.04405 [quant-ph]} \BibitemShut {NoStop}%
\bibitem [{\citenamefont {Dall’Arno}\ \emph {et~al.}(2017)\citenamefont {Dall’Arno}, \citenamefont {Brandsen},\ and\ \citenamefont {Buscemi}}]{Dall}%
  \BibitemOpen
  \bibfield  {author} {\bibinfo {author} {\bibfnamefont {M.}~\bibnamefont {Dall’Arno}}, \bibinfo {author} {\bibfnamefont {S.}~\bibnamefont {Brandsen}},\ and\ \bibinfo {author} {\bibfnamefont {F.}~\bibnamefont {Buscemi}},\ }\bibfield  {title} {\bibinfo {title} {Device-independent tests of quantum channels},\ }\href {https://doi.org/10.1098/rspa.2016.0721} {\bibfield  {journal} {\bibinfo  {journal} {Proceedings of the Royal Society A: Mathematical, Physical and Engineering Sciences}\ }\textbf {\bibinfo {volume} {473}},\ \bibinfo {pages} {20160721} (\bibinfo {year} {2017})}\BibitemShut {NoStop}%
\bibitem [{\citenamefont {Sekatski}\ \emph {et~al.}(2018)\citenamefont {Sekatski}, \citenamefont {Bancal}, \citenamefont {Wagner},\ and\ \citenamefont {Sangouard}}]{pavel}%
  \BibitemOpen
  \bibfield  {author} {\bibinfo {author} {\bibfnamefont {P.}~\bibnamefont {Sekatski}}, \bibinfo {author} {\bibfnamefont {J.-D.}\ \bibnamefont {Bancal}}, \bibinfo {author} {\bibfnamefont {S.}~\bibnamefont {Wagner}},\ and\ \bibinfo {author} {\bibfnamefont {N.}~\bibnamefont {Sangouard}},\ }\bibfield  {title} {\bibinfo {title} {Certifying the building blocks of quantum computers from bell's theorem},\ }\href {https://doi.org/10.1103/PhysRevLett.121.180505} {\bibfield  {journal} {\bibinfo  {journal} {Phys. Rev. Lett.}\ }\textbf {\bibinfo {volume} {121}},\ \bibinfo {pages} {180505} (\bibinfo {year} {2018})}\BibitemShut {NoStop}%
\bibitem [{\citenamefont {Sarkar}(2024)}]{sarkarint}%
  \BibitemOpen
  \bibfield  {author} {\bibinfo {author} {\bibfnamefont {S.}~\bibnamefont {Sarkar}},\ }\bibfield  {title} {\bibinfo {title} {Model-independent inference of quantum interaction from statistics},\ }\bibfield  {journal} {\bibinfo  {journal} {Physical Review A}\ }\textbf {\bibinfo {volume} {110}},\ \href {https://doi.org/10.1103/physreva.110.l020402} {10.1103/physreva.110.l020402} (\bibinfo {year} {2024})\BibitemShut {NoStop}%
\bibitem [{\citenamefont {Reichardt}\ \emph {et~al.}(2013{\natexlab{b}})\citenamefont {Reichardt}, \citenamefont {Unger},\ and\ \citenamefont {Vazirani}}]{Reichardt2013}%
  \BibitemOpen
  \bibfield  {author} {\bibinfo {author} {\bibfnamefont {B.~W.}\ \bibnamefont {Reichardt}}, \bibinfo {author} {\bibfnamefont {F.}~\bibnamefont {Unger}},\ and\ \bibinfo {author} {\bibfnamefont {U.}~\bibnamefont {Vazirani}},\ }\bibfield  {title} {\bibinfo {title} {Classical command of quantum systems},\ }\href {https://doi.org/10.1038/nature12035} {\bibfield  {journal} {\bibinfo  {journal} {Nature}\ }\textbf {\bibinfo {volume} {496}},\ \bibinfo {pages} {456} (\bibinfo {year} {2013}{\natexlab{b}})}\BibitemShut {NoStop}%
\bibitem [{\citenamefont {Fitzsimons}\ \emph {et~al.}(2018)\citenamefont {Fitzsimons}, \citenamefont {Hajdušek},\ and\ \citenamefont {Morimae}}]{Fitzsimons_2018}%
  \BibitemOpen
  \bibfield  {author} {\bibinfo {author} {\bibfnamefont {J.~F.}\ \bibnamefont {Fitzsimons}}, \bibinfo {author} {\bibfnamefont {M.}~\bibnamefont {Hajdušek}},\ and\ \bibinfo {author} {\bibfnamefont {T.}~\bibnamefont {Morimae}},\ }\bibfield  {title} {\bibinfo {title} {Post hoc verification of quantum computation},\ }\bibfield  {journal} {\bibinfo  {journal} {Physical Review Letters}\ }\textbf {\bibinfo {volume} {120}},\ \href {https://doi.org/10.1103/physrevlett.120.040501} {10.1103/physrevlett.120.040501} (\bibinfo {year} {2018})\BibitemShut {NoStop}%
\bibitem [{\citenamefont {Sarkar}\ \emph {et~al.}(2023)\citenamefont {Sarkar}, \citenamefont {Borka\l{}a}, \citenamefont {Jebarathinam}, \citenamefont {Makuta}, \citenamefont {Saha},\ and\ \citenamefont {Augusiak}}]{sarkar3}%
  \BibitemOpen
  \bibfield  {author} {\bibinfo {author} {\bibfnamefont {S.}~\bibnamefont {Sarkar}}, \bibinfo {author} {\bibfnamefont {J.~J.}\ \bibnamefont {Borka\l{}a}}, \bibinfo {author} {\bibfnamefont {C.}~\bibnamefont {Jebarathinam}}, \bibinfo {author} {\bibfnamefont {O.}~\bibnamefont {Makuta}}, \bibinfo {author} {\bibfnamefont {D.}~\bibnamefont {Saha}},\ and\ \bibinfo {author} {\bibfnamefont {R.}~\bibnamefont {Augusiak}},\ }\bibfield  {title} {\bibinfo {title} {Self-testing of any pure entangled state with the minimal number of measurements and optimal randomness certification in a one-sided device-independent scenario},\ }\href {https://doi.org/10.1103/PhysRevApplied.19.034038} {\bibfield  {journal} {\bibinfo  {journal} {Phys. Rev. Appl.}\ }\textbf {\bibinfo {volume} {19}},\ \bibinfo {pages} {034038} (\bibinfo {year} {2023})}\BibitemShut {NoStop}%
\bibitem [{\citenamefont {Sarkar}\ \emph {et~al.}(2024{\natexlab{b}})\citenamefont {Sarkar}, \citenamefont {Alexandre C.~Orthey}, \citenamefont {Sharma},\ and\ \citenamefont {Augusiak}}]{sarkar2024steer}%
  \BibitemOpen
  \bibfield  {author} {\bibinfo {author} {\bibfnamefont {S.}~\bibnamefont {Sarkar}}, \bibinfo {author} {\bibfnamefont {J.}~\bibnamefont {Alexandre C.~Orthey}}, \bibinfo {author} {\bibfnamefont {G.}~\bibnamefont {Sharma}},\ and\ \bibinfo {author} {\bibfnamefont {R.}~\bibnamefont {Augusiak}},\ }\href {https://arxiv.org/abs/2402.18522} {\bibinfo {title} {Almost device-independent certification of gme states with minimal measurements}} (\bibinfo {year} {2024}{\natexlab{b}}),\ \Eprint {https://arxiv.org/abs/2402.18522} {arXiv:2402.18522 [quant-ph]} \BibitemShut {NoStop}%
\end{thebibliography}
\providecommand{\noopsort}[1]{}\providecommand{\singleletter}[1]{#1}%

\onecolumngrid
\appendix

\begin{center}
   \textbf{Supplementary Material: Any gate of a quantum computer can be certified device-independently}
\end{center}

\section{Almost DI certification}


Let us begin by stating the self-testing of sources, $A_i's$ and $L's$ measurement from \cite{sarkar2024universal} when $E$ chooses $e=0$, that is, does not interact with the incoming systems.
\setcounter{thm}{0}
\begin{fakt}\label{theorem1}
Consider the quantum network in Fig. \ref{fig1} with $e=0$ and assume that the Bell functional $\langle\mathcal{I}_l\rangle$ \eqref{BE1Nm} for any $l$ attain the quantum bound and each outcome of $L$ occurs with probability $p(l)=1/2^N$ where $N$ denotes the number of external parties. The sources $P_i$ prepare the states $\ket{\psi_{A_i\overline{A}_i}}\in \mathcal{H}_{A_i}\otimes\mathcal{H}_{\overline{A}_i}$ for $i=1,2,\ldots,N$, the measurement of $L$ is given by $\{M_l\}$ for $l=l_1l_2\ldots l_N$ such that $l_i=0,1$ which acts on $\bigotimes_{i=1}^N\mathcal{H}_{\overline{A}_i}$, the  local observables for each of the party is given by $\mathcal{A}_{i,0}$, $\mathcal{A}_{i,1}$, and $\mathcal{A}_{i,2}$ for $i=1,2,\ldots,N$ which act on $\mathcal{H}_{A_i}$. Then,
\begin{enumerate}
    \item The Hilbert spaces of all the parties decompose as $\mathcal{H}_{A_i}=\mathcal{H}_{A_i'}\otimes\mathcal{H}_{A_i''}$ and $\mathcal{H}_{\overline{A}_i}=\mathcal{H}_{\overline{A}_i'}\otimes\mathcal{H}_{\overline{A}_i''}$, where $\mathcal{H}_{A_i'}$ and 
    $\mathcal{H}_{\overline{A}_i'}$ are qubit Hilbert spaces, whereas $\mathcal{H}_{A_i''}$ and $\mathcal{H}_{\overline{A}_i''}$
    are some finite-dimensional but unknown auxiliary Hilbert spaces.
    \item There exist local unitary transformations $U_{A_i}:\mathcal{H}_{A_i}\rightarrow(\mathbb{C}^2)_{A'}\otimes\mathcal{H}_{A''_i}$ and $U_{\overline{A}_i}:\mathcal{H}_{\overline{A}_i}\rightarrow(\mathbb{C}^2)_{A'}\otimes \mathcal{H}_{L''_i}$ 
    such that the states are given by
\begin{eqnarray}\label{A10}
U_{A_i}\otimes U_{\overline{A}_i}\ket{\psi_{A_i\overline{A}_i}}=\ket{\phi^+_{A_i'\overline{A}'_i}}\otimes\ket{\xi_{A_i''\overline{A}''_i}}.
\end{eqnarray}
\item The measurement of $L$ is
\begin{eqnarray}\label{statest1}
 U_LM_lU_L^{\dagger} =\proj{\phi_l}_{L'}\otimes\I_{L''}\qquad \forall l,
\end{eqnarray}
where $U_L=U_{\overline{A}_1}\otimes\ldots\otimes U_{\overline{A}_N}$ and $\ket{\phi_l}$ are given in Eq. \eqref{GHZvecsm}. The measurements of all the other parties are given by
\begin{eqnarray}\label{mea1}
U_{A_1}\mathcal{A}_{1,0}\,U_{A_1}^{\dagger}&=&\left(\frac{X+Z}{\sqrt{2}}\right)_{\mathcal{A}_1'}\otimes\I_{A_1''}, \quad U_{A_1}A_{1,1}\,U_{\mathcal{A}_1}^{\dagger}=\left(\frac{X-Z}{\sqrt{2}}\right)_{A_1'}\otimes\I_{A_1''},\nonumber\\
U_{A_i}\mathcal{A}_{i,0}\,U_{A_i}^{\dagger}&=&Z_{A_i'}\otimes\I_{A_i''},\quad\quad\ \ \ \ \  U_{A_i}\mathcal{A}_{i,1}\,U_{A_i}^{\dagger}=X_{A_i'}\otimes\I_{A_i''}\qquad (i=2,\ldots,N)\label{stmea1}
\end{eqnarray}
and,
\begin{eqnarray}\label{mea2}
    U_{A_i}\mathcal{A}_{i,2}\,U_{A_i}^{\dagger}&=&\pm Y_{A_i'}\otimes\I_{A_i''} \qquad\forall i.\label{stmea2}
\end{eqnarray}
\end{enumerate}
\end{fakt}

\begin{proof}
    The proof of the above fact can be found in \cite{sarkar2024universal}.
\end{proof}

\begin{thm}\label{theorem3}
Suppose that the correlations generated in the quantum network [Fig. \ref{fig1}] satisfy the conditions to obtain Theorem \ref{theorem1} along with the additional statistics in Eq. (\ref{betstatfullm}). Then, we can conclude (under the assumption that $V$ does not change $\bigotimes_{i=1}^N\mathcal{H}_{\overline{A}_i}$) that 
\begin{eqnarray}
    \label{stmeaNLWE1}
U_L V U^{\dagger}_L &=&\mathcal{U}^* \otimes \mathbbm{1}_{L''},\quad \mathrm{if}\quad A_{i,2}=Y\otimes\I \quad\forall i, \quad\text{or}\quad \nonumber\\
 U_L V U^{\dagger}_L &=& \mathcal{U} \otimes \mathbbm{1}_{L''}\quad \ \ \ \mathrm{if}\quad  A_{i,2}=-Y\otimes\I \quad\forall i
\end{eqnarray}
where $U_L$ is the same unitary as in Theorem \ref{theorem1}
and $L''=\overline{A_1}''\ldots\overline{A_N}''$.
\end{thm}
\begin{proof}
Let us begin by expanding \eqref{betstatfullm} as
\begin{eqnarray}\label{cond1}
\sum_{i_1,\ldots,i_N=0,1,2,3}f_{l,i_1,\ldots,i_N}\left\langle U_{A}\tilde{\mathcal{A}}_{1,i_1}\bigotimes_{k=2}^N \mathcal{A}_{k,i_k}U_{A}^{\dagger}\otimes\tilde{V}^{\dagger}_{L}\proj{\phi_l}_{L'}\otimes\I_{L''}\tilde{V}_{L}\right\rangle_{\ket{\phi^+_{A'L'}}\otimes\ket{\xi_{A''L''}}} =\frac{1}{2^N}
\end{eqnarray}
where $\tilde{V}_{L}=U_{L}V U_{L}^{\dagger}$ and $U_A=\bigotimes_{i=1}^{N}U_{A_i}$. Expressing the above formula as
 \begin{equation}
\sum_{i_1,\ldots,i_N=0,1,2,3}f_{l,i_1,\ldots,i_N}\left\langle U_{A}\tilde{\mathcal{A}}_{1,i_1}\bigotimes_{k=2}^N \mathcal{A}_{k,i_k}U_{A}^{\dagger} \otimes U_L R_l U_L^{\dagger}\right\rangle_{\ket{\phi^+_{A'L'}}\otimes\ket{\xi_{A''L''}}} =\frac{1}{2^N}
\end{equation}
where $R_l=V^{\dagger}M_lV$. Notice that it represents the exact condition as Theorem 2 of \cite{sarkar2024universal}. Consequently, we can conclude from here that if fact \eqref{theorem1} is satisfied for $e=0$, then 
\begin{eqnarray}\label{stmeaNLWE1p}
U_L (V^{\dagger}M_lV) U^{\dagger}_L &=&\proj{\delta_l^*} \otimes \mathbbm{1}_{L''},\quad \mathrm{if}\quad A_{i,2}=Y\otimes\I \quad\forall i, \quad\text{or}\quad \nonumber\\
 U_L (V^{\dagger}M_lV) U_L^{\dagger} &=& \proj{\delta_l} \otimes \mathbbm{1}_{L''}\quad \ \ \ \mathrm{if}\quad  A_{i,2}=-Y\otimes\I \quad\forall i
\end{eqnarray}
We recall from fact \eqref{theorem1} that $M_l$ is certified on $\bigotimes_{i=1}^N\mathcal{H}_{\overline{A}_i}$ to be \eqref{mea1}. Since, $V$ does not change the local supports, that is, the Hilbert space $\mathcal{H}_L=\bigotimes_{i=1}^N\mathcal{H}_{\overline{A}_i}$ remains invariant under action of $V$ we obtain for all $l$
 \begin{equation}
     U_L V^{\dagger} U_L^{\dagger}(\proj{\phi_l}\otimes\I_{L''})U_LVU_{L}^{\dagger}=\proj{\delta_l}\otimes\I_{L''}
 \end{equation}
Expressing $\tilde{V}=U_L V^{\dagger} U_L^{\dagger}$, the above formula can also be expressed as
\begin{equation}
    \tilde{V}\otimes\I_\mathcal{L''}(\proj{\phi_l}\otimes\proj{\phi^+_D}_{L''\mathcal{L''}})\tilde{V}^{\dagger}\otimes\I_\mathcal{L''}=\proj{\delta_l}\otimes\proj{\phi^+_D}_{L''\mathcal{L''}}
\end{equation}
where we introduced an auxiliary Hilbert space $\mathcal{H}_{\mathcal{L''}}$ with $D$ denoting the dimension of $\mathcal{H}_{L''}$. Furthermore, the above expression also allows us to conclude that
\begin{equation}\label{eq5}
    \tilde{V}\otimes\I_\mathcal{L''}(\ket{\phi_l}\otimes\ket{\phi^+_D}_{L''\mathcal{L''}})=\ket{\delta_l}\otimes\ket{\phi^+_D}_{L''\mathcal{L''}}
\end{equation}
Now, any operator $\tilde{V}$ acting on $(\mathbb{C}^{2})^{\otimes N}\otimes\mathcal{H}_{L''}$ can be expressed as
\begin{equation}\label{genV}
    \tilde{V}=\sum_{i,j=0}^{2^N-1}\ket{\phi_i}\!\bra{\phi_j}\otimes F_{ij}
\end{equation}
where $F_{ij}$ are some operators acting on $\mathcal{H}_{L''}$. Putting $\tilde{V}$ back into Eq. \eqref{eq5}, we obtain
\begin{equation}
    \sum_{i=0}^{2^N-1}\ket{\phi_i}\otimes (F_{il}\otimes\I_\mathcal{L''}\ket{\phi^+_D}_{L''\mathcal{L''}})=\ket{\delta_l}\otimes\ket{\phi^+_D}_{L''\mathcal{L''}}.
\end{equation}
Now, multiplying the above expression with $\bra{\phi_i}$, we obtain
\begin{equation}
F_{il}\otimes\I_\mathcal{L''}\ket{\phi^+_D}_{L''\mathcal{L''}}=\bra{\phi_i}\delta_l\rangle\ket{\phi^+_D}_{L''\mathcal{L''}}.
\end{equation}
Multiplying the above expression with $\bra{\phi^+_D}_{L''\mathcal{L''}}$ from the right-hand side and taking a partial trace over $\mathcal{L''}$, we obtain
\begin{equation}
    F_{il}=\bra{\phi_i}\delta_l\rangle\I.
\end{equation}
Putting it back into Eq. \eqref{genV}, we obtain
\begin{equation}
     \tilde{V}=U_L V^{\dagger} U_L^{\dagger}=\sum_{i,j=0}^{2^N-1}\ket{\phi_i}\!\bra{\phi_i}\delta_j\rangle\!\bra{\phi_j}\otimes \I=\sum_{j=0}^{2^N-1}|\delta_j\rangle\!\bra{\phi_j}\otimes \I.
\end{equation}
Similarly, one can repeat the same steps with $\ket{\delta_l^*}$ and recalling that $\ket{\phi_l^*}=\ket{\phi_l}$ to obtain \eqref{stmeaNLWE1p}. This completes the proof.
    
\end{proof}

\section{DI certification}

Let us again begin by stating the self-testing of sources, $A_i's, E's$ and $L's$ measurement from \cite{sarkar2024universal} when $E$ chooses $e=0$, that is, does not transform the subsystem but only does a joint measurement on them. The theorem is divided in two parts, where the first concerns self-testing $E's$ measurements $\{\mathcal{R}_{i,r_i}\}$ where $r\equiv r_{i,1}r_{i,2}\ (r_{i,1},r_{i,2}=0,1)$ along with $A's, L's$ observables $\mathcal{A}_{i,x_i},\mathcal{B}_{i,y_i}$ for $i=1,\ldots,N$ and $x_i,y_i=0,1$, while the second part concerns self-testing the remaining $L's$ joint measurement $\{M_l\}$ and $A's$ observables  $\mathcal{A}_{i,2}$ for all $i$.
\setcounter{thm}{0}
\begin{fakt}\label{theorem1}
{\it{Part 1.}}
Consider the quantum network in Fig. \ref{fig2} with $e=0$ and consider the Bell functional $\langle\mathcal{K}_{r_i,i}\rangle_{\ket{\psi_{A_iR_{i,1}}}\ket{\psi_{R_{i,2}\overline{A}_i}}}$, similar to $\langle\mathcal{I}_{l}\rangle$ Eq. \eqref{BE1Nm} when $N=2$, constructed using the probabilities $p(a_i,r_i,l_i|x_i,e=0,y_i)$ as
\begin{eqnarray}
    \mathcal{K}_{r_1,1}= (-1)^{r_{i,1}}\tilde{\mathcal{A}}_{1,1} \mathcal{B}_{1,1} +(-1)^{r_{i,2}}\tilde{\mathcal{A}}_{1,0}\mathcal{B}_{1,0} 
\end{eqnarray}
where $ \tilde{\mathcal{A}}_{1,0}=\frac{(\mathcal{A}_{1,0}-\mathcal{A}_{1,1})}{\sqrt{2}},\ \tilde{\mathcal{A}}_{1,1}=\frac{(\mathcal{A}_{1,0}+\mathcal{A}_{1,1})}{\sqrt{2}}$ and for $i=2,\ldots,N$
\begin{eqnarray}
    \mathcal{K}_{r_i,i}= (-1)^{r_{i,1}}\mathcal{A}_{i,1} \tilde{\mathcal{B}}_{i,1} +(-1)^{r_{i,2}}\mathcal{A}_{i,0}\tilde{\mathcal{B}}_{i,0} 
\end{eqnarray}
where $ \tilde{\mathcal{B}}_{i,0}=\frac{(\mathcal{B}_{i,0}-\mathcal{B}_{i,1})}{\sqrt{2}},\ \tilde{\mathcal{B}}_{1,1}=\frac{(\mathcal{B}_{i,0}+\mathcal{B}_{i,1})}{\sqrt{2}}$. Suppose now that these Bell inequalities are maximally violated among $A_i's$ and $L$ along with Eve obtaining each outcome $r_i$ with a probability $1/4$. The sources $P_i, P_i'$ prepare the states $\ket{\psi_{A_iR_{i,1}}}\in \mathcal{H}_{A_i}\otimes\mathcal{H}_{R_{i,1}}, \ket{\psi_{R_{i,2}\overline{A}_i}}\in \mathcal{H}_{R_{i,2}}\otimes\mathcal{H}_{\overline{A}_i}$ respectively for any $i$. Then,
\begin{enumerate}
    \item The Hilbert spaces of all the parties decompose as $\mathcal{H}_{A_i}=\mathcal{H}_{A_i'}\otimes\mathcal{H}_{A_i''}$, $\mathcal{H}_{R_{i,j}}=\mathcal{H}_{R_{i,j}'}\otimes\mathcal{H}_{R_{i,j}''} (j=1,2)$, and $\mathcal{H}_{\overline{A}_i}=\mathcal{H}_{\overline{A}_i'}\otimes\mathcal{H}_{\overline{A}_i''}$, where $\mathcal{H}_{A_i'}, \mathcal{H}_{R_{i,j}'}$ and 
    $\mathcal{H}_{\overline{A}_i'}$ are qubit Hilbert spaces, whereas $\mathcal{H}_{A_i''}, \mathcal{H}_{R_{i,j}''}$ and $\mathcal{H}_{\overline{A}_i''}$
    are some finite-dimensional but unknown auxiliary Hilbert spaces.
    \item There exist local unitary transformations $U_{A_i}:\mathcal{H}_{A_i}\rightarrow(\mathbb{C}^2)_{A'}\otimes\mathcal{H}_{A''_i}$, $U_{A_i}:\mathcal{H}_{R_{i,j}}\rightarrow(\mathbb{C}^2)_{R_{i,j}'}\otimes\mathcal{H}_{R_{i,j}''}$, and $U_{\overline{A}_i}:\mathcal{H}_{\overline{A}_i}\rightarrow(\mathbb{C}^2)_{\overline{A}_i'}\otimes \mathcal{H}_{\overline{A}_i''}$ 
    such that the states are given by
\begin{eqnarray}\label{A10}
U_{A_i}\otimes U_{R_{i,1}}\ket{\psi_{A_iR_{i,1}}}=\ket{\phi^+_{A_i'R_{i,1}'}}\otimes\ket{\xi_{A_i''R_{i,1}''}},\nonumber\\
U_{R_{i,2}}\otimes U_{\overline{A}_i}\ket{\psi_{R_{i,2}\overline{A}_i}}=\ket{\phi^+_{R_{i,2}'\overline{A}'_i}}\otimes\ket{\xi_{R_{i,2}''\overline{A}''_i}}.
\end{eqnarray}
\item The measurement at $E$ is certified on the support of the space $\bigotimes_{j=1,2}(\mathbb{C}^2)_{R_{i,j}'}\otimes\mathcal{H}_{R_{i,j}''}$ as
\begin{eqnarray}\label{statest1}
 U_{R_{i,1}}\otimes U_{R_{i,2}} \mathcal{R}_{i,r_i}U_{R_{i,1}}^{\dagger}\otimes U_{R_{i,2}}^{\dagger} =\proj{\phi_{r_i}^{(2)}}_{R_{1,i}'R_{2,i}'}\otimes\I_{R_{1,i}''R_{2,i}''},
\end{eqnarray}
where $\ket{\phi_{r_i}^{(2)}}$ are the orthonormal vectors stated in \eqref{GHZvecsm} for $N=2$. The measurements of all the other parties are given by
\begin{eqnarray}\label{mea111}
U_{A_1}\mathcal{A}_{1,0}\,U_{A_1}^{\dagger}&=&\left(\frac{X+Z}{\sqrt{2}}\right)_{A_1'}\otimes\I_{A_1''}, \quad U_{A_1}A_{1,1}\,U_{\mathcal{A}_1}^{\dagger}=\left(\frac{X-Z}{\sqrt{2}}\right)_{A_1'}\otimes\I_{A_1''},\nonumber\\
U_{A_i}\mathcal{A}_{i,0}\,U_{A_i}^{\dagger}&=&Z_{A_i'}\otimes\I_{A_i''},\quad\quad\ \ \ \ \  U_{A_i}\mathcal{A}_{i,1}\,U_{A_i}^{\dagger}=X_{A_i'}\otimes\I_{A_i''}\qquad (i=2,\ldots,N)\label{stmea1}
\end{eqnarray}
and,
\begin{eqnarray}\label{mea11}
U_{\overline{A}_i}\mathcal{B}_{i,0}\,U_{\overline{A}_i}^{\dagger}&=&\left(\frac{X+Z}{\sqrt{2}}\right)_{\overline{A}_i'}\otimes\I_{\overline{A}_i''}, \quad U_{\overline{A}_i}\mathcal{B}_{i,1}\,U_{\overline{A}_i}^{\dagger}=\left(\frac{X-Z}{\sqrt{2}}\right)_{\overline{A}_i'}\otimes\I_{\overline{A}_i''}, \qquad (i=2,\ldots,N) \nonumber\\
U_{\overline{A}_i}\mathcal{B}_{1,0}\,U_{\overline{A}_i}^{\dagger}&=&Z_{\overline{A}_i'}\otimes\I_{\overline{A}_i''},\quad\quad\ \ \ \ \  U_{\overline{A}_i}\mathcal{B}_{1,1}\,U_{\overline{A}_i}^{\dagger}=X_{\overline{A}_1'}\otimes\I_{\overline{A}_1''}
\end{eqnarray}

{\it{Part 2.}} Consider now the statistics corresponding to the case when measurement at $E$ outputs $r_i=0$ for all $i$. Assume that the Bell functional constructed out of these statistics $\langle\mathcal{I}_l\rangle$ \eqref{BE1Nm} for any outcome $l$ at $L$ is maximally violated among $A_i's$ with each outcome of $L$ occurring with probability $p(l)=1/2^N$ where $N$ denotes the number of external parties. Then, the joint measurement of $L$, $\{M_l\}$ is certified as
\begin{eqnarray}\label{statest11}
 U_LM_lU_L^{\dagger} =\proj{\phi_l}_{L'}\otimes\I_{L''}\qquad \forall l,
\end{eqnarray}
where $U_L=U_{\overline{A}_1}\otimes\ldots\otimes U_{\overline{A}_N}$ and $\ket{\phi_l}$ are given in Eq. \eqref{GHZvecsm},
and,
\begin{eqnarray}\label{mea21}
    U_{A_i}\mathcal{A}_{i,2}\,U_{A_i}^{\dagger}&=&\pm Y_{A_i'}\otimes\I_{A_i''} \qquad\forall i.\label{stmea2}
\end{eqnarray}
\end{enumerate}
\end{fakt}

\begin{proof}
    The proof of {\it{Part 1.}} follows from Theorem 1 Ref. \cite{sarkar2024universal} for $N=2$, and the proof of {\it{Part 2.}} is the same as in Theorem 1 Ref. \cite{sarkar2024universal}.
\end{proof}

\setcounter{thm}{1}
\begin{thm}\label{theorem4}
Suppose that the correlations generated in the quantum network [Fig. \ref{fig2}] satisfy the conditions to obtain Theorem \ref{theorem1} along with the additional statistics in Eq. (\ref{betstatfullmn1}). Then, we can conclude that 
\begin{eqnarray}
    \label{stmeaNLWE112}
U_{R_1} \overline{V} U^{\dagger}_{R_1} &=&\mathcal{U}^* \otimes \mathbbm{1}_{{R_1}''},\quad \mathrm{if}\quad A_{i,2}=Y\otimes\I \quad\forall i, \quad\text{or}\quad \nonumber\\
 U_{R_1} \overline{V} U^{\dagger}_{R_1} &=& \mathcal{U} \otimes \mathbbm{1}_{{R_1}''}\quad \ \ \ \mathrm{if}\quad  A_{i,2}=-Y\otimes\I \quad\forall i
\end{eqnarray}
where $U_E=\bigotimes_{i=1}^NU_{R_{i,1}}$ where $U_{R_{i,1}}$ is the same unitaries as specified in Theorem \ref{theorem1}
and ${R_1}''=R_{1,1}''\ldots R_{N,1}''$. Here $\overline{V}$ is the projection of $V$ onto the space $\bigotimes_{i=1}^N\mathcal{H}_{R_{i,1}}$.
\end{thm}

\begin{proof}

Let us consider the condition \eqref{betstatfullm12121} and expand it to get
\begin{equation}\label{betstatfullmn1}
\sum_{i_1,\ldots,i_N=0}^3f_{l,i_1,\ldots,i_N}{\bra{\psi_{AR_1}}V^{\dagger}_{R_1}\bra{\psi_{R_2\overline{A}}}}\tilde{\mathcal{A}}_{1,i_1}\bigotimes_{k=2}^N \mathcal{A}_{k,i_k} \otimes\mathcal{R}_0\otimes M_l\ket{\psi_{R_2\overline{A}}}V_{R_1}{\ket{\psi_{AR_1}}}=\frac{\overline{p}(0)}{2^N}.
\end{equation}
Now,$\ket{\psi_{R_2\overline{A}}}=\bigotimes_{i=1}^N{\ket{\psi_{R_{2,i}\overline{A}_i}}},{\ket{\psi_{AR_1}}}=\bigotimes_{i=1}^N{\ket{\psi_{A_iR_{1,i}}}}$ are certified from Theorem \ref{theorem1} for $N=2$ as
\begin{eqnarray}\label{ststate2}
    U_{R_{2,i}}\otimes U_{\overline{A}_i}\ket{\psi_{R_{2,i}\overline{A}_i}}=\ket{\phi^+_{R_{2,i}'\overline{A}_i'}}\otimes\ket{\xi_{R_{2,i}''\overline{A}_i''}}\nonumber\\
 U_{A_i}\otimes  U_{R_{1,i}}\ket{\psi_{A_iR_{1,i}}}=\ket{\phi^+_{A_i'R_{1,i}'}}\otimes\ket{\xi_{A_i''R_{1,i}''}}
\end{eqnarray}
and, $\mathcal{R}_l=\bigotimes_{i=1}^N\mathcal{R}_{i,r_i}$ is certified as
\begin{eqnarray}
 \overline{\mathcal{R}}_{i,r_i}=   U_{R_{1,i}}\otimes U_{R_{2,i}}\mathcal{R}_{i,r_i} U_{R_{1,i}}^{\dagger}\otimes U_{R_{2,i}}^{\dagger}=\proj{\phi_{r_i}^{(2)}}_{R_{1,i}'R_{2,i}'}\otimes\I_{R_{1,i}''R_{2,i}''}\oplus J_{i,r_i}
\end{eqnarray}
where $\ket{\phi_{r_i}^{(2)}}$ are the orthonormal vectors stated in \eqref{GHZvecsm} for $N=2$. Here, $J_{i,r_i}$ which is a positive operator denotes the part of the uncharacterized measurement as $\mathcal{R}_i$ could only be characterised on the support of $\Tr_{A_i}(\psi_{A_iR_{2,i}})\otimes\Tr_{\overline{A}_i}(\psi_{R_{2,i}\overline{A}_i})$ such that $J_{l,i}$ is orthogonal to this support. Now, suppose a state $\ket{\gamma}$ acts on $ \overline{\mathcal{R}}_{i,r_i}$ for any $i$. However, the only part of the state that can get teleported to $L$ is the projected state $\prod_{R_{1,i}'R_{1,i}''}\ket{\gamma}$ where $\prod_{R_{1,i}'R_{1,i}''}$ is the projection onto the Hilbert space $\mathcal{H}_{R_{1,i}'}\otimes\mathcal{H}_{R_{1,i}''}$. Consequently, the condition \eqref{betstatfullmn1} can be rewritten as
\begin{equation}\label{betstatfullmn122}
\sum_{i_1,\ldots,i_N=0}^3f_{l,i_1,\ldots,i_N}{\bra{\psi_{AR_1}}\overline{V}^{\dagger}_{R_1}\bra{\psi_{R_2\overline{A}}}}\tilde{\mathcal{A}}_{1,i_1}\bigotimes_{k=2}^N \mathcal{A}_{k,i_k} \otimes\mathcal{R}_0\otimes M_l\ket{\psi_{R_2\overline{A}}}\overline{V}_{R_1}{\ket{\psi_{AR_1}}}=\frac{\overline{p}(0)}{2^N}
\end{equation}
where $\overline{V}_{R_1}=\prod_{R_1'R_1''}V\prod_{R_1'R_1''}$.

Now, $M_l$ is also certified from Theorem \ref{theorem1} as Eq. \eqref{statest1}. We now recall that $\bigotimes_{i=1}^N\ket{\phi^+_{R_{2,i}'\overline{A}_i'}}=\ket{\phi^{+,2^N}_{R_2,\overline{A}}}$ where $\ket{\phi^{+,2^N}}$ is the maximally entangled state of local dimension $2^N$. Consequently, we have that 
\begin{eqnarray}
\Tr_{R_2\overline{A}}(U_{R_1}\mathcal{R}_0U_{R_1}^{\dagger}\otimes M_l\proj{\psi_{R_2\overline{A}}})=\overline{p}(0)\proj{\phi_l}_{R_1'}\otimes\I_{R_1''}
\end{eqnarray}
where $U_{R_1}=\bigotimes_{i=1}^N U_{R_{1,i}}$.
Putting it back into \eqref{betstatfullmn122}, we obtain
\begin{eqnarray}
    \sum_{i_1,\ldots,i_N=0}^3f_{l,i_1,\ldots,i_N}{\bra{\psi_{AR_1}}}\overline{V}^{\dagger}_{R_1}\tilde{\mathcal{A}}_{1,i_1}\bigotimes_{k=2}^N \mathcal{A}_{k,i_k} \otimes U_{R_1}^{\dagger}\proj{\phi_l}_{R_1'}\otimes\I_{R_1''}U_{R_1}\overline{V}_{R_1}{\ket{\psi_{AR_1}}}=\frac{1}{2^N}.
\end{eqnarray}
Putting in $\ket{\psi_{AR_1}}$ from \eqref{ststate2}, we obtain
\begin{eqnarray}
\sum_{i_1,\ldots,i_N=0}^3f_{l,i_1,\ldots,i_N}\left\langle U_{A}\tilde{\mathcal{A}}_{1,i_1}\bigotimes_{k=2}^N \mathcal{A}_{k,i_k} U_{A}^{\dagger}\otimes \tilde{V}^{\dagger}_{R_1}\proj{\phi_l}_{R_1'}\otimes\I_{R_1''}\tilde{V}_{R_1}\right\rangle_{\ket{\phi^+_{A'R_1'}}\otimes\ket{\xi_{A''R_1''}}}=\frac{1}{2^N}
\end{eqnarray}
where $\tilde{V}_{R_1}=U_{R_1}\overline{V}_{R_1}U_{R_1}^{\dagger}$ and $U_A=\bigotimes_{i=1}^{N}U_{A_i}$. This is the same condition as \eqref{cond1}. Thus, we can straightaway follow the proof of Theorem \ref{theorem3} to obtain \eqref{stmeaNLWE112}. This completes the proof.

\end{proof}

\end{document}